\documentclass[11pt]{article}
\usepackage[a4paper]{geometry}
\usepackage{amsfonts, amsmath, amssymb, amsthm, graphicx, caption, authblk, multirow, makecell, framed, float, xcolor, enumitem, tikz, hyperref}
\usepackage{algorithm,algorithmicx,algpseudocode}
\usetikzlibrary{fit,shapes.arrows}
\theoremstyle{definition}
\newtheorem{theorem}{Theorem}

\newtheorem{lemma}{Lemma}

\theoremstyle{remark}

\definecolor{blk}{RGB}{63,63,63}
\newcommand*{\mybox}[1]{%
  \framebox{\raisebox{0cm}[0.5\baselineskip][0.05\baselineskip]{%
    \hbox to 0.10cm {\hss#1\hss}}}\hspace{0.05cm}}

\begin{document}
\title{NP-Completeness and Physical Zero-Knowledge Proofs for Zeiger}
\author[1]{Suthee Ruangwises\thanks{\texttt{suthee@cp.eng.chula.ac.th}}}
\affil[1]{Department of Computer Engineering, Chulalongkorn University, Bangkok, Thailand}
\date{}
\maketitle

\begin{abstract}
Zeiger is a pencil puzzle consisting of a rectangular grid, with each cell having an arrow pointing in horizontal or vertical direction. Some cells also contain a positive integer. The objective of this puzzle is to fill a positive integer into every unnumbered cell such that the integer in each cell is equal to the number of different integers in all cells along the direction an arrow in that cell points to. In this paper, we prove that deciding solvability of a given Zeiger puzzle is NP-complete via a reduction from the not-all-equal positive 3SAT (NAE3SAT+) problem. We also construct a card-based physical zero-knowledge proof protocol for Zeiger, which enables a prover to physically show a verifier the existence of the puzzle's solution without revealing it.

\end{abstract}

\section{Introduction}
\textit{Zeiger} (also known as \textit{Number Pointers} or \textit{Arrows}) is a pencil puzzle regularly published in a German magazine Denksel \cite{janko}. This puzzle also appears in many other places online, including mobile apps \cite{apps}. The puzzle consists of a $k \times \ell$ rectangular grid, with each cell having an arrow pointing in horizontal or vertical direction. Some cells also contain a positive integer. The player has to fill a positive integer into every unnumbered cell such that the integer in each cell is equal to the number of different integers in all cells along the direction an arrow in that cell points to (see Fig. \ref{fig1}).

\begin{figure}
\centering
\begin{tikzpicture}
\draw[step=0.8cm, line width=0.3mm] (0,0) grid (4.0,4.0);

\node[single arrow, draw=black, single arrow head extend=4pt, minimum height=5.5mm, minimum width=7mm, rotate=0] at (0.4,0.4) {};
\node[single arrow, draw=black, single arrow head extend=4pt, minimum height=5.5mm, minimum width=7mm, rotate=90] at (1.2,0.4) {};
\node[single arrow, draw=black, single arrow head extend=4pt, minimum height=5.5mm, minimum width=7mm, rotate=90] at (2.0,0.4) {};
\node[single arrow, draw=black, single arrow head extend=4pt, minimum height=5.5mm, minimum width=7mm, rotate=180] at (2.8,0.4) {};
\node[single arrow, draw=black, single arrow head extend=4pt, minimum height=5.5mm, minimum width=7mm, rotate=180] at (3.6,0.4) {};

\node[single arrow, draw=black, single arrow head extend=4pt, minimum height=5.5mm, minimum width=7mm, rotate=90] at (0.4,1.2) {};
\node[single arrow, draw=black, single arrow head extend=4pt, minimum height=5.5mm, minimum width=7mm, rotate=0] at (1.2,1.2) {};
\node[single arrow, draw=black, single arrow head extend=4pt, minimum height=5.5mm, minimum width=7mm, rotate=270] at (2.0,1.2) {};
\node[single arrow, draw=black, single arrow head extend=4pt, minimum height=5.5mm, minimum width=7mm, rotate=90] at (2.8,1.2) {};
\node[single arrow, draw=black, single arrow head extend=4pt, minimum height=5.5mm, minimum width=7mm, rotate=90] at (3.6,1.2) {};

\node[single arrow, draw=black, single arrow head extend=4pt, minimum height=5.5mm, minimum width=7mm, rotate=90] at (0.4,2.0) {};
\node[single arrow, draw=black, single arrow head extend=4pt, minimum height=5.5mm, minimum width=7mm, rotate=270] at (1.2,2.0) {};
\node[single arrow, draw=black, single arrow head extend=4pt, minimum height=5.5mm, minimum width=7mm, rotate=0] at (2.0,2.0) {};
\node[single arrow, draw=black, single arrow head extend=4pt, minimum height=5.5mm, minimum width=7mm, rotate=180] at (2.8,2.0) {};
\node[single arrow, draw=black, single arrow head extend=4pt, minimum height=5.5mm, minimum width=7mm, rotate=90] at (3.6,2.0) {};

\node[single arrow, draw=black, single arrow head extend=4pt, minimum height=5.5mm, minimum width=7mm, rotate=0] at (0.4,2.8) {};
\node[single arrow, draw=black, single arrow head extend=4pt, minimum height=5.5mm, minimum width=7mm, rotate=0] at (1.2,2.8) {};
\node[single arrow, draw=black, single arrow head extend=4pt, minimum height=5.5mm, minimum width=7mm, rotate=90] at (2.0,2.8) {};
\node[single arrow, draw=black, single arrow head extend=4pt, minimum height=5.5mm, minimum width=7mm, rotate=270] at (2.8,2.8) {};
\node[single arrow, draw=black, single arrow head extend=4pt, minimum height=5.5mm, minimum width=7mm, rotate=180] at (3.6,2.8) {};

\node[single arrow, draw=black, single arrow head extend=4pt, minimum height=5.5mm, minimum width=7mm, rotate=270] at (0.4,3.6) {};
\node[single arrow, draw=black, single arrow head extend=4pt, minimum height=5.5mm, minimum width=7mm, rotate=270] at (1.2,3.6) {};
\node[single arrow, draw=black, single arrow head extend=4pt, minimum height=5.5mm, minimum width=7mm, rotate=270] at (2.0,3.6) {};
\node[single arrow, draw=black, single arrow head extend=4pt, minimum height=5.5mm, minimum width=7mm, rotate=180] at (2.8,3.6) {};
\node[single arrow, draw=black, single arrow head extend=4pt, minimum height=5.5mm, minimum width=7mm, rotate=270] at (3.6,3.6) {};

\node at (2.0,0.4) {2};
\node at (2.8,2.0) {1};
\node at (3.6,2.8) {3};
\node at (2.8,3.6) {2};
\end{tikzpicture}
\hspace{1.2cm}
\begin{tikzpicture}
\draw[step=0.8cm, line width=0.3mm] (0,0) grid (4.0,4.0);

\node[single arrow, draw=black, single arrow head extend=4pt, minimum height=5.5mm, minimum width=7mm, rotate=0] at (0.4,0.4) {};
\node[single arrow, draw=black, single arrow head extend=4pt, minimum height=5.5mm, minimum width=7mm, rotate=90] at (1.2,0.4) {};
\node[single arrow, draw=black, single arrow head extend=4pt, minimum height=5.5mm, minimum width=7mm, rotate=90] at (2.0,0.4) {};
\node[single arrow, draw=black, single arrow head extend=4pt, minimum height=5.5mm, minimum width=7mm, rotate=180] at (2.8,0.4) {};
\node[single arrow, draw=black, single arrow head extend=4pt, minimum height=5.5mm, minimum width=7mm, rotate=180] at (3.6,0.4) {};

\node[single arrow, draw=black, single arrow head extend=4pt, minimum height=5.5mm, minimum width=7mm, rotate=90] at (0.4,1.2) {};
\node[single arrow, draw=black, single arrow head extend=4pt, minimum height=5.5mm, minimum width=7mm, rotate=0] at (1.2,1.2) {};
\node[single arrow, draw=black, single arrow head extend=4pt, minimum height=5.5mm, minimum width=7mm, rotate=270] at (2.0,1.2) {};
\node[single arrow, draw=black, single arrow head extend=4pt, minimum height=5.5mm, minimum width=7mm, rotate=90] at (2.8,1.2) {};
\node[single arrow, draw=black, single arrow head extend=4pt, minimum height=5.5mm, minimum width=7mm, rotate=90] at (3.6,1.2) {};

\node[single arrow, draw=black, single arrow head extend=4pt, minimum height=5.5mm, minimum width=7mm, rotate=90] at (0.4,2.0) {};
\node[single arrow, draw=black, single arrow head extend=4pt, minimum height=5.5mm, minimum width=7mm, rotate=270] at (1.2,2.0) {};
\node[single arrow, draw=black, single arrow head extend=4pt, minimum height=5.5mm, minimum width=7mm, rotate=0] at (2.0,2.0) {};
\node[single arrow, draw=black, single arrow head extend=4pt, minimum height=5.5mm, minimum width=7mm, rotate=180] at (2.8,2.0) {};
\node[single arrow, draw=black, single arrow head extend=4pt, minimum height=5.5mm, minimum width=7mm, rotate=90] at (3.6,2.0) {};

\node[single arrow, draw=black, single arrow head extend=4pt, minimum height=5.5mm, minimum width=7mm, rotate=0] at (0.4,2.8) {};
\node[single arrow, draw=black, single arrow head extend=4pt, minimum height=5.5mm, minimum width=7mm, rotate=0] at (1.2,2.8) {};
\node[single arrow, draw=black, single arrow head extend=4pt, minimum height=5.5mm, minimum width=7mm, rotate=90] at (2.0,2.8) {};
\node[single arrow, draw=black, single arrow head extend=4pt, minimum height=5.5mm, minimum width=7mm, rotate=270] at (2.8,2.8) {};
\node[single arrow, draw=black, single arrow head extend=4pt, minimum height=5.5mm, minimum width=7mm, rotate=180] at (3.6,2.8) {};

\node[single arrow, draw=black, single arrow head extend=4pt, minimum height=5.5mm, minimum width=7mm, rotate=270] at (0.4,3.6) {};
\node[single arrow, draw=black, single arrow head extend=4pt, minimum height=5.5mm, minimum width=7mm, rotate=270] at (1.2,3.6) {};
\node[single arrow, draw=black, single arrow head extend=4pt, minimum height=5.5mm, minimum width=7mm, rotate=270] at (2.0,3.6) {};
\node[single arrow, draw=black, single arrow head extend=4pt, minimum height=5.5mm, minimum width=7mm, rotate=180] at (2.8,3.6) {};
\node[single arrow, draw=black, single arrow head extend=4pt, minimum height=5.5mm, minimum width=7mm, rotate=270] at (3.6,3.6) {};

\node at (0.4,0.4) {2};
\node at (1.2,0.4) {3};
\node at (2.0,0.4) {2};
\node at (2.8,0.4) {2};
\node at (3.6,0.4) {2};

\node at (0.4,1.2) {2};
\node at (1.2,1.2) {3};
\node at (2.0,1.2) {1};
\node at (2.8,1.2) {3};
\node at (3.6,1.2) {2};

\node at (0.4,2.0) {1};
\node at (1.2,2.0) {1};
\node at (2.0,2.0) {1};
\node at (2.8,2.0) {1};
\node at (3.6,2.0) {1};

\node at (0.4,2.8) {3};
\node at (1.2,2.8) {2};
\node at (2.0,2.8) {1};
\node at (2.8,2.8) {3};
\node at (3.6,2.8) {3};

\node at (0.4,3.6) {3};
\node at (1.2,3.6) {3};
\node at (2.0,3.6) {2};
\node at (2.8,3.6) {2};
\node at (3.6,3.6) {3};
\end{tikzpicture}
\caption{An example of a $5 \times 5$ Zeiger puzzle (left) and its solution (right)}
\label{fig1}
\end{figure}

Solving Zeiger is difficult, as the puzzle's constraints involve relationships between multiple cells. Yuki Fujimoto, the developer of a Zeiger mobile app as well as many other pencil puzzle apps, stated that Zeiger is the most difficult pencil puzzle \cite{apps}. Unsurprisingly, the puzzle turns out to be NP-complete, as we will later prove in this paper.

Suppose Agnes, an expert in Zeiger, created a Zeiger puzzle and challenged her friend Brian to solve it. After several tries, he could not solve the puzzle and doubted whether it has a solution. Agnes needs to convince him that her puzzle has a solution \textit{without} revealing the solution itself (which would make the challenge pointless). In this situation, she needs a \textit{zero-knowledge proof (ZKP)}.

\subsection{Zero-Knowledge Proof}
A ZKP is an interactive proof between a prover $P$ and a verifier $V$. Both $P$ and $V$ are given a computational problem $x$, but only $P$ knows its solution $w$. A ZKP with perfect completeness and soundness must satisfy the following three properties.

\begin{enumerate}
	\item \textbf{Perfect Completeness:} If $P$ knows $w$, then $V$ always accepts.
	\item \textbf{Perfect Soundness:} If $P$ does not know $w$, then $V$ always rejects.
	\item \textbf{Zero-knowledge:} $V$ learns nothing about $w$. Formally, there exists a probabilistic polynomial time algorithm $S$ (called a \textit{simulator}) that does not know $w$ but has access to $V$, and the outputs of $S$ follow the same probability distribution as the ones of the real protocol.
\end{enumerate}

The concept of a ZKP was introduced in 1989 by Goldwasser et al. \cite{zkp0}. Although a ZKP exists for every NP problem \cite{zkp}, it is more reasonable to construct ZKPs for NP-complete problems, as $V$ cannot compute the solution by themselves in such problems. As one would expect, many popular pencil puzzles have been proved to be NP-complete, including Bridges \cite{npbridges}, Five Cells \cite{npfivecells}, Goishi Hiroi \cite{npgoishi}, Heyawake \cite{npheyawake}, Kakuro \cite{npsudoku}, Makaro \cite{npmakaro}, Nondango \cite{npnondango}, Nonogram \cite{npnonogram}, Numberlink \cite{npnumberlink}, Nurikabe \cite{npnurikabe}, Ripple Effect \cite{npripple}, Shikaku \cite{npripple}, Slitherlink \cite{npsudoku}, Sudoku \cite{npsudoku}, and Sumplete \cite{sumplete}.

\subsection{Physical Zero-Knowledge Proof Protocols}
Instead of simply constructing computational ZKPs for these puzzles via reductions, many researchers have developed physical ZKPs using a deck of playing cards. These card-based protocols have the benefit that they do not require computers and also allow external observers to verify that the prover truthfully executes them (which is a challenging task for digital protocols). In addition, they are more intuitive and easier to verify the correctness and security, even for non-experts, and thus can be used for didactic purpose.

There is a line of work dedicated to constructing card-based physical ZKP protocols for pencil puzzles, such as Five Cells \cite{decom}, Goishi Hiroi \cite{goishi}, Heyawake \cite{nurikabe}, Kakuro \cite{kakuro}, Makaro \cite{makaro}, Masyu \cite{slitherlink}, Meadows \cite{decom}, Nonogram \cite{nonogram}, Numberlink \cite{numberlink}, Nurikabe \cite{nurikabe}, Ripple Effect \cite{ripple}, Shikaku \cite{shikaku}, Slitherlink \cite{slitherlink}, Sudoku \cite{sudoku}, Sumplete \cite{sumplete}, and Toichika \cite{goishi}.

\subsection{Our Contribution}
In this paper, we prove that deciding solvability of a given Zeiger puzzle is NP-complete via a reduction from the not-all-equal positive 3SAT (NAE3SAT+) problem. We also construct a physical ZKP protocol for Zeiger using a deck of playing cards.

\section{NP-Completeness Proof of Zeiger}
In this section, we will prove that deciding whether a given Zeiger puzzle has a solution is NP-complete.

As the problem clearly belongs to NP, the nontrivial part is to prove the NP-hardness. We will do so by constructing a reduction from the not-all-equal positive 3SAT (NAE3SAT+) problem. The NAE3SAT+ is the problem of deciding whether there exists a Boolean assignment such that every clause has at least one true literal and at least one false literal, in a setting where each clause contains exactly three positive literals (see Fig. \ref{fig2}). This problem is known to be NP-complete \cite{sat}.

\begin{figure}
\centering
\begin{tikzpicture}
\node at (0,0) {$C_1 = x_1 \vee x_2 \vee x_3$};
\node at (0,-0.5) {$C_2 = x_2 \vee x_3 \vee x_5$};
\node at (0,-1) {$C_3 = x_1 \vee x_4 \vee x_5$};
\node at (0,-1.5) {$C_4 = x_2 \vee x_4 \vee x_5$};
\node at (0,-2) {};
\end{tikzpicture}
\hspace{3cm}
\begin{tikzpicture}
\node at (0,0) {$x_1 = \text{TRUE}$};
\node at (0,-0.5) {$x_2 = \text{FALSE}$};
\node at (0,-1) {$x_3 = \text{TRUE}$};
\node at (0,-1.5) {$x_4 = \text{TRUE}$};
\node at (0,-2) {$x_5 = \text{FALSE}$};
\end{tikzpicture}
\caption{An NAE3SAT+ instance (left) and one of its possible solutions (right)}
\label{fig2}
\end{figure}

Suppose we are given an NAE3SAT+ instance $S$ with $m$ clauses and $n$ variables. Let $C_1,C_2,...,C_m$ be the clauses of $S$ and $x_1,x_2,...,x_n$ be the variables of $S$. We assume that every variable appears in at least one clause (otherwise we can just remove unused variables).

We will transform $S$ into a Zeiger puzzle grid $G$ with size $(m+3) \times (n+5)$. The intuition is that each of the first $m$ rows of $G$ corresponds to each clause of $S$, each of the first $n$ columns of $G$ corresponds to each variable of $S$, and numbers 2 and 3 correspond to TRUE and FALSE, respectively.

Let $a(p,q)$ denotes the cell in the $p$-th row and $q$-th column of $G$. We define $a(p,q)$ as follows.

\begin{minipage}{\linewidth}
\begin{itemize}
	\item If $q \leq n$,
	\begin{itemize}
		\item if $p \leq m$,
		\begin{itemize}
			\item if $x_q$ appears in $C_p$, then $a(p,q)$ contains a down arrow with no number;
			\item otherwise, $a(p,q)$ contains a right arrow with a number 4;
		\end{itemize}
		\item if $p=m+1$, then $a(p,q)$ contains a right arrow with a number 4;
		\item if $p=m+2$, then $a(p,q)$ contains an up arrow with a number 2;
		\item if $p=m+3$, then $a(p,q)$ contains an up arrow with no number.
	\end{itemize}
	\item If $q=n+1$, then $a(p,q)$ contains a right arrow with a number 4.
	\item If $q=n+2$,
	\begin{itemize}
		\item If $p \leq m$, then $a(p,q)$ contains a left arrow with a number 3;
		\item otherwise, $a(p,q)$ contains a right arrow with a number 3.
	\end{itemize}
	\item If $q=n+3$, then $a(p,q)$ contains a right arrow with a number 2.
	\item If $q=n+4$, then $a(p,q)$ contains a right arrow with a number 1.
	\item If $q=n+5$, then $a(p,q)$ contains a left arrow with a number 4 (see Fig. \ref{fig3}).
\end{itemize}
\end{minipage}

\begin{figure}
\centering
\begin{tikzpicture}
\draw[step=0.8cm, line width=0.3mm] (0,0) grid (8.0,5.6);

\node[single arrow, draw=black, single arrow head extend=4pt, minimum height=5.5mm, minimum width=7mm, rotate=90] at (0.4,0.4) {};
\node[single arrow, draw=black, single arrow head extend=4pt, minimum height=5.5mm, minimum width=7mm, rotate=90] at (0.4,1.2) {};
\node[single arrow, draw=black, single arrow head extend=4pt, minimum height=5.5mm, minimum width=7mm, rotate=0] at (0.4,2.0) {};
\node[single arrow, draw=black, single arrow head extend=4pt, minimum height=5.5mm, minimum width=7mm, rotate=0] at (0.4,2.8) {};
\node[single arrow, draw=black, single arrow head extend=4pt, minimum height=5.5mm, minimum width=7mm, rotate=270] at (0.4,3.6) {};
\node[single arrow, draw=black, single arrow head extend=4pt, minimum height=5.5mm, minimum width=7mm, rotate=0] at (0.4,4.4) {};
\node[single arrow, draw=black, single arrow head extend=4pt, minimum height=5.5mm, minimum width=7mm, rotate=270] at (0.4,5.2) {};

\node[single arrow, draw=black, single arrow head extend=4pt, minimum height=5.5mm, minimum width=7mm, rotate=90] at (1.2,0.4) {};
\node[single arrow, draw=black, single arrow head extend=4pt, minimum height=5.5mm, minimum width=7mm, rotate=90] at (1.2,1.2) {};
\node[single arrow, draw=black, single arrow head extend=4pt, minimum height=5.5mm, minimum width=7mm, rotate=0] at (1.2,2.0) {};
\node[single arrow, draw=black, single arrow head extend=4pt, minimum height=5.5mm, minimum width=7mm, rotate=270] at (1.2,2.8) {};
\node[single arrow, draw=black, single arrow head extend=4pt, minimum height=5.5mm, minimum width=7mm, rotate=0] at (1.2,3.6) {};
\node[single arrow, draw=black, single arrow head extend=4pt, minimum height=5.5mm, minimum width=7mm, rotate=270] at (1.2,4.4) {};
\node[single arrow, draw=black, single arrow head extend=4pt, minimum height=5.5mm, minimum width=7mm, rotate=270] at (1.2,5.2) {};

\node[single arrow, draw=black, single arrow head extend=4pt, minimum height=5.5mm, minimum width=7mm, rotate=90] at (2.0,0.4) {};
\node[single arrow, draw=black, single arrow head extend=4pt, minimum height=5.5mm, minimum width=7mm, rotate=90] at (2.0,1.2) {};
\node[single arrow, draw=black, single arrow head extend=4pt, minimum height=5.5mm, minimum width=7mm, rotate=0] at (2.0,2.0) {};
\node[single arrow, draw=black, single arrow head extend=4pt, minimum height=5.5mm, minimum width=7mm, rotate=0] at (2.0,2.8) {};
\node[single arrow, draw=black, single arrow head extend=4pt, minimum height=5.5mm, minimum width=7mm, rotate=0] at (2.0,3.6) {};
\node[single arrow, draw=black, single arrow head extend=4pt, minimum height=5.5mm, minimum width=7mm, rotate=270] at (2.0,4.4) {};
\node[single arrow, draw=black, single arrow head extend=4pt, minimum height=5.5mm, minimum width=7mm, rotate=270] at (2.0,5.2) {};

\node[single arrow, draw=black, single arrow head extend=4pt, minimum height=5.5mm, minimum width=7mm, rotate=90] at (2.8,0.4) {};
\node[single arrow, draw=black, single arrow head extend=4pt, minimum height=5.5mm, minimum width=7mm, rotate=90] at (2.8,1.2) {};
\node[single arrow, draw=black, single arrow head extend=4pt, minimum height=5.5mm, minimum width=7mm, rotate=0] at (2.8,2.0) {};
\node[single arrow, draw=black, single arrow head extend=4pt, minimum height=5.5mm, minimum width=7mm, rotate=270] at (2.8,2.8) {};
\node[single arrow, draw=black, single arrow head extend=4pt, minimum height=5.5mm, minimum width=7mm, rotate=270] at (2.8,3.6) {};
\node[single arrow, draw=black, single arrow head extend=4pt, minimum height=5.5mm, minimum width=7mm, rotate=0] at (2.8,4.4) {};
\node[single arrow, draw=black, single arrow head extend=4pt, minimum height=5.5mm, minimum width=7mm, rotate=0] at (2.8,5.2) {};

\node[single arrow, draw=black, single arrow head extend=4pt, minimum height=5.5mm, minimum width=7mm, rotate=90] at (3.6,0.4) {};
\node[single arrow, draw=black, single arrow head extend=4pt, minimum height=5.5mm, minimum width=7mm, rotate=90] at (3.6,1.2) {};
\node[single arrow, draw=black, single arrow head extend=4pt, minimum height=5.5mm, minimum width=7mm, rotate=0] at (3.6,2.0) {};
\node[single arrow, draw=black, single arrow head extend=4pt, minimum height=5.5mm, minimum width=7mm, rotate=270] at (3.6,2.8) {};
\node[single arrow, draw=black, single arrow head extend=4pt, minimum height=5.5mm, minimum width=7mm, rotate=270] at (3.6,3.6) {};
\node[single arrow, draw=black, single arrow head extend=4pt, minimum height=5.5mm, minimum width=7mm, rotate=270] at (3.6,4.4) {};
\node[single arrow, draw=black, single arrow head extend=4pt, minimum height=5.5mm, minimum width=7mm, rotate=0] at (3.6,5.2) {};

\node[single arrow, draw=black, single arrow head extend=4pt, minimum height=5.5mm, minimum width=7mm, rotate=0] at (4.4,0.4) {};
\node[single arrow, draw=black, single arrow head extend=4pt, minimum height=5.5mm, minimum width=7mm, rotate=0] at (4.4,1.2) {};
\node[single arrow, draw=black, single arrow head extend=4pt, minimum height=5.5mm, minimum width=7mm, rotate=0] at (4.4,2.0) {};
\node[single arrow, draw=black, single arrow head extend=4pt, minimum height=5.5mm, minimum width=7mm, rotate=0] at (4.4,2.8) {};
\node[single arrow, draw=black, single arrow head extend=4pt, minimum height=5.5mm, minimum width=7mm, rotate=0] at (4.4,3.6) {};
\node[single arrow, draw=black, single arrow head extend=4pt, minimum height=5.5mm, minimum width=7mm, rotate=0] at (4.4,4.4) {};
\node[single arrow, draw=black, single arrow head extend=4pt, minimum height=5.5mm, minimum width=7mm, rotate=0] at (4.4,5.2) {};

\node[single arrow, draw=black, single arrow head extend=4pt, minimum height=5.5mm, minimum width=7mm, rotate=0] at (5.2,0.4) {};
\node[single arrow, draw=black, single arrow head extend=4pt, minimum height=5.5mm, minimum width=7mm, rotate=0] at (5.2,1.2) {};
\node[single arrow, draw=black, single arrow head extend=4pt, minimum height=5.5mm, minimum width=7mm, rotate=0] at (5.2,2.0) {};
\node[single arrow, draw=black, single arrow head extend=4pt, minimum height=5.5mm, minimum width=7mm, rotate=180] at (5.2,2.8) {};
\node[single arrow, draw=black, single arrow head extend=4pt, minimum height=5.5mm, minimum width=7mm, rotate=180] at (5.2,3.6) {};
\node[single arrow, draw=black, single arrow head extend=4pt, minimum height=5.5mm, minimum width=7mm, rotate=180] at (5.2,4.4) {};
\node[single arrow, draw=black, single arrow head extend=4pt, minimum height=5.5mm, minimum width=7mm, rotate=180] at (5.2,5.2) {};

\node[single arrow, draw=black, single arrow head extend=4pt, minimum height=5.5mm, minimum width=7mm, rotate=0] at (6.0,0.4) {};
\node[single arrow, draw=black, single arrow head extend=4pt, minimum height=5.5mm, minimum width=7mm, rotate=0] at (6.0,1.2) {};
\node[single arrow, draw=black, single arrow head extend=4pt, minimum height=5.5mm, minimum width=7mm, rotate=0] at (6.0,2.0) {};
\node[single arrow, draw=black, single arrow head extend=4pt, minimum height=5.5mm, minimum width=7mm, rotate=0] at (6.0,2.8) {};
\node[single arrow, draw=black, single arrow head extend=4pt, minimum height=5.5mm, minimum width=7mm, rotate=0] at (6.0,3.6) {};
\node[single arrow, draw=black, single arrow head extend=4pt, minimum height=5.5mm, minimum width=7mm, rotate=0] at (6.0,4.4) {};
\node[single arrow, draw=black, single arrow head extend=4pt, minimum height=5.5mm, minimum width=7mm, rotate=0] at (6.0,5.2) {};

\node[single arrow, draw=black, single arrow head extend=4pt, minimum height=5.5mm, minimum width=7mm, rotate=0] at (6.8,0.4) {};
\node[single arrow, draw=black, single arrow head extend=4pt, minimum height=5.5mm, minimum width=7mm, rotate=0] at (6.8,1.2) {};
\node[single arrow, draw=black, single arrow head extend=4pt, minimum height=5.5mm, minimum width=7mm, rotate=0] at (6.8,2.0) {};
\node[single arrow, draw=black, single arrow head extend=4pt, minimum height=5.5mm, minimum width=7mm, rotate=0] at (6.8,2.8) {};
\node[single arrow, draw=black, single arrow head extend=4pt, minimum height=5.5mm, minimum width=7mm, rotate=0] at (6.8,3.6) {};
\node[single arrow, draw=black, single arrow head extend=4pt, minimum height=5.5mm, minimum width=7mm, rotate=0] at (6.8,4.4) {};
\node[single arrow, draw=black, single arrow head extend=4pt, minimum height=5.5mm, minimum width=7mm, rotate=0] at (6.8,5.2) {};

\node[single arrow, draw=black, single arrow head extend=4pt, minimum height=5.5mm, minimum width=7mm, rotate=180] at (7.6,0.4) {};
\node[single arrow, draw=black, single arrow head extend=4pt, minimum height=5.5mm, minimum width=7mm, rotate=180] at (7.6,1.2) {};
\node[single arrow, draw=black, single arrow head extend=4pt, minimum height=5.5mm, minimum width=7mm, rotate=180] at (7.6,2.0) {};
\node[single arrow, draw=black, single arrow head extend=4pt, minimum height=5.5mm, minimum width=7mm, rotate=180] at (7.6,2.8) {};
\node[single arrow, draw=black, single arrow head extend=4pt, minimum height=5.5mm, minimum width=7mm, rotate=180] at (7.6,3.6) {};
\node[single arrow, draw=black, single arrow head extend=4pt, minimum height=5.5mm, minimum width=7mm, rotate=180] at (7.6,4.4) {};
\node[single arrow, draw=black, single arrow head extend=4pt, minimum height=5.5mm, minimum width=7mm, rotate=180] at (7.6,5.2) {};

\node at (0.4,0.4) {};
\node at (0.4,1.2) {2};
\node at (0.4,2.0) {4};
\node at (0.4,2.8) {4};
\node at (0.4,3.6) {};
\node at (0.4,4.4) {4};
\node at (0.4,5.2) {};

\node at (1.2,0.4) {};
\node at (1.2,1.2) {2};
\node at (1.2,2.0) {4};
\node at (1.2,2.8) {};
\node at (1.2,3.6) {4};
\node at (1.2,4.4) {};
\node at (1.2,5.2) {};

\node at (2.0,0.4) {};
\node at (2.0,1.2) {2};
\node at (2.0,2.0) {4};
\node at (2.0,2.8) {4};
\node at (2.0,3.6) {4};
\node at (2.0,4.4) {};
\node at (2.0,5.2) {};

\node at (2.8,0.4) {};
\node at (2.8,1.2) {2};
\node at (2.8,2.0) {4};
\node at (2.8,2.8) {};
\node at (2.8,3.6) {};
\node at (2.8,4.4) {4};
\node at (2.8,5.2) {4};

\node at (3.6,0.4) {};
\node at (3.6,1.2) {2};
\node at (3.6,2.0) {4};
\node at (3.6,2.8) {};
\node at (3.6,3.6) {};
\node at (3.6,4.4) {};
\node at (3.6,5.2) {4};

\node at (4.4,0.4) {4};
\node at (4.4,1.2) {4};
\node at (4.4,2.0) {4};
\node at (4.4,2.8) {4};
\node at (4.4,3.6) {4};
\node at (4.4,4.4) {4};
\node at (4.4,5.2) {4};

\node at (5.2,0.4) {3};
\node at (5.2,1.2) {3};
\node at (5.2,2.0) {3};
\node at (5.2,2.8) {3};
\node at (5.2,3.6) {3};
\node at (5.2,4.4) {3};
\node at (5.2,5.2) {3};

\node at (6.0,0.4) {2};
\node at (6.0,1.2) {2};
\node at (6.0,2.0) {2};
\node at (6.0,2.8) {2};
\node at (6.0,3.6) {2};
\node at (6.0,4.4) {2};
\node at (6.0,5.2) {2};

\node at (6.8,0.4) {1};
\node at (6.8,1.2) {1};
\node at (6.8,2.0) {1};
\node at (6.8,2.8) {1};
\node at (6.8,3.6) {1};
\node at (6.8,4.4) {1};
\node at (6.8,5.2) {1};

\node at (7.6,0.4) {4};
\node at (7.6,1.2) {4};
\node at (7.6,2.0) {4};
\node at (7.6,2.8) {4};
\node at (7.6,3.6) {4};
\node at (7.6,4.4) {4};
\node at (7.6,5.2) {4};
\end{tikzpicture}
\caption{A $7 \times 10$ Zeiger puzzle grid transformed from the NAE3SAT+ instance in Fig.~\ref{fig2}}
\label{fig3}
\end{figure}

Clearly, this transformation can be done in polynomial time.

Consider each of the first $n$ columns of $G$. We disregard constraints on right arrows in this column for now, and consider only the constraints on up and down arrows in this column. Next, we will prove that there are only two ways to fill numbers into this column to satisfy them: filling either all 2s or all 3s.

\begin{lemma} \label{lemA}
For each $q \leq n$, there are exactly two ways to fill numbers into all unnumbered cells in the $q$-th column of $G$ to satisfy constraints of all up and down arrows in that column, which are filling all 2s and filling all 3s.
\end{lemma}

\begin{proof}
Let $a(p,q)$ be the bottommost down arrow cell in the $q$-th column of $G$ (there must be at least one down arrow since every variable appears in at least one clause). Observe that each cell below $a(p,q)$ contains either a 2 or 4 (and there is at least one 2 and at least one 4), except the cell $a(m+3,q)$ which currently contains no number, so $a(p,q)$ must contain either a 2 or 3.

\textbf{Case 1:} $a(p,q)$ contains a 2.

Since $a(m+2,q)$ contains a 2, all cells above it must contain either a 2 or 4 (and there is at least one 2 and at least one 4), which implies $a(m+3,q)$ also contains a 2. Consider the second bottommost down arrow in this column. All cells below it contains either a 2 or 4 (and there is at least one 2 and at least one 4), so it must contain a 2. We then consider in the same way the third bottommost down arrow, the fourth bottommost down arrow, and so on. Finally, we can conclude that every down arrow in this column must contain a 2.

\textbf{Case 2:} $a(p,q)$ contains a 3.

Since $a(m+2,q)$ contains a 2, all cells above it must contain either a 3 or 4 (and there is at least one 3 and at least one 4), which implies $a(m+3,q)$ contains a 3. Consider the second bottommost down arrow in this column. All cells below it contains either a 2, 3, or 4 (and there is at least one number of each type), so it must contain a 3. We then consider in the same way the third bottommost down arrow, the fourth bottommost down arrow, and so on. Finally, we can conclude that every down arrow in this column must contain a 3.

On the other hand, filling all 2s or all 3s clearly satisfies the constraints of all up and down arrows in the $q$-th column, so the converse is also true.
\end{proof}

Now consider each $p$-th row of $G$ ($p \leq m$). Since $C_p$ contains exactly three literals, there are exactly three unnumbered down arrow cells in the $p$-th row. We will prove that they must be filled with at least one 2 and at least one 3.

\begin{lemma} \label{lemB}
In any solution of $G$, the three unnumbered down arrow cells in each $p$-th row ($p \leq m$) must be filled with at least one 2 and at least one 3.
\end{lemma}

\begin{proof}
Consider the three unnumbered down arrow cells in the $p$-th row. From Lemma \ref{lemA}, each of them must be filled with either a 2 or 3. Also, as $a(p,n+2)$ contains a left arrow with a number 3, exactly three different numbers must appear among the cells to the left of $a(p,n+2)$. Besides these three cells, all other such cells contain all 4s (and there is at least one 4 in $a(p,n+1)$). Therefore, these three cells must contain at least one 2 and at least one 3.
\end{proof}

Finally, we will prove the following theorem, which implies NP-hardness of the puzzle.

\begin{theorem} \label{thm1}
$S$ has a solution if and only if $G$ has a solution.
\end{theorem}

\begin{proof}
Suppose $S$ has a solution $T$. From $T$, we construct a solution $H$ of $G$ as follows: for each $q \leq n$, we fill every unnumbered cell in the $q$-th column of $G$ with a 2 if $x_q$ is assigned to TRUE in $T$, and with a 3 otherwise (see Fig. \ref{fig4}).

From Lemma \ref{lemA}, the filled numbers satisfy constraints on all up and down arrows in the first $n$ columns of $G$. Also, since each clause of $S$ has at least one true literal and at least one false literal, the three unnumbered cells in each $p$-th row of $G$ ($p \leq m$) are filled with at least one 2 and at least one 3, which satisfies the constraint on the arrow in $a(p,n+2)$. It is easy to verify that constraints on other arrows in $G$ are also satisfied. Hence, we can conclude that $H$ is a valid solution of $G$.

On the other hand, suppose $G$ has a solution $H'$. From $H'$, we construct a solution $T'$ of $S$ as follows: for each $q \leq n$, we assign $x_q$ to TRUE if $a(m+3,q)$ contains a 2 in $H'$, and to FALSE otherwise.

From Lemma \ref{lemA}, all unnumbered cells in each $q$-th column of $G$ ($q \leq n$) must be filled with the same number as the one in $a(m+3,q)$. From Lemma \ref{lemB}, the three unnumbered cells in each $p$-th row of $G$ ($p \leq m$) must be filled with at least one 2 and at least one 3, which means the three literals in $C_p$ of $S$ must contain at least one true literal and at least one false literal. Since this holds for every $p \leq m$, we can conclude that $T'$ is a valid solution of $S$.
\end{proof}

\begin{figure}
\centering
\begin{tikzpicture}
\draw[step=0.8cm, line width=0.3mm] (0,0) grid (8.0,5.6);

\node[single arrow, draw=black, single arrow head extend=4pt, minimum height=5.5mm, minimum width=7mm, rotate=90] at (0.4,0.4) {};
\node[single arrow, draw=black, single arrow head extend=4pt, minimum height=5.5mm, minimum width=7mm, rotate=90] at (0.4,1.2) {};
\node[single arrow, draw=black, single arrow head extend=4pt, minimum height=5.5mm, minimum width=7mm, rotate=0] at (0.4,2.0) {};
\node[single arrow, draw=black, single arrow head extend=4pt, minimum height=5.5mm, minimum width=7mm, rotate=0] at (0.4,2.8) {};
\node[single arrow, draw=black, single arrow head extend=4pt, minimum height=5.5mm, minimum width=7mm, rotate=270] at (0.4,3.6) {};
\node[single arrow, draw=black, single arrow head extend=4pt, minimum height=5.5mm, minimum width=7mm, rotate=0] at (0.4,4.4) {};
\node[single arrow, draw=black, single arrow head extend=4pt, minimum height=5.5mm, minimum width=7mm, rotate=270] at (0.4,5.2) {};

\node[single arrow, draw=black, single arrow head extend=4pt, minimum height=5.5mm, minimum width=7mm, rotate=90] at (1.2,0.4) {};
\node[single arrow, draw=black, single arrow head extend=4pt, minimum height=5.5mm, minimum width=7mm, rotate=90] at (1.2,1.2) {};
\node[single arrow, draw=black, single arrow head extend=4pt, minimum height=5.5mm, minimum width=7mm, rotate=0] at (1.2,2.0) {};
\node[single arrow, draw=black, single arrow head extend=4pt, minimum height=5.5mm, minimum width=7mm, rotate=270] at (1.2,2.8) {};
\node[single arrow, draw=black, single arrow head extend=4pt, minimum height=5.5mm, minimum width=7mm, rotate=0] at (1.2,3.6) {};
\node[single arrow, draw=black, single arrow head extend=4pt, minimum height=5.5mm, minimum width=7mm, rotate=270] at (1.2,4.4) {};
\node[single arrow, draw=black, single arrow head extend=4pt, minimum height=5.5mm, minimum width=7mm, rotate=270] at (1.2,5.2) {};

\node[single arrow, draw=black, single arrow head extend=4pt, minimum height=5.5mm, minimum width=7mm, rotate=90] at (2.0,0.4) {};
\node[single arrow, draw=black, single arrow head extend=4pt, minimum height=5.5mm, minimum width=7mm, rotate=90] at (2.0,1.2) {};
\node[single arrow, draw=black, single arrow head extend=4pt, minimum height=5.5mm, minimum width=7mm, rotate=0] at (2.0,2.0) {};
\node[single arrow, draw=black, single arrow head extend=4pt, minimum height=5.5mm, minimum width=7mm, rotate=0] at (2.0,2.8) {};
\node[single arrow, draw=black, single arrow head extend=4pt, minimum height=5.5mm, minimum width=7mm, rotate=0] at (2.0,3.6) {};
\node[single arrow, draw=black, single arrow head extend=4pt, minimum height=5.5mm, minimum width=7mm, rotate=270] at (2.0,4.4) {};
\node[single arrow, draw=black, single arrow head extend=4pt, minimum height=5.5mm, minimum width=7mm, rotate=270] at (2.0,5.2) {};

\node[single arrow, draw=black, single arrow head extend=4pt, minimum height=5.5mm, minimum width=7mm, rotate=90] at (2.8,0.4) {};
\node[single arrow, draw=black, single arrow head extend=4pt, minimum height=5.5mm, minimum width=7mm, rotate=90] at (2.8,1.2) {};
\node[single arrow, draw=black, single arrow head extend=4pt, minimum height=5.5mm, minimum width=7mm, rotate=0] at (2.8,2.0) {};
\node[single arrow, draw=black, single arrow head extend=4pt, minimum height=5.5mm, minimum width=7mm, rotate=270] at (2.8,2.8) {};
\node[single arrow, draw=black, single arrow head extend=4pt, minimum height=5.5mm, minimum width=7mm, rotate=270] at (2.8,3.6) {};
\node[single arrow, draw=black, single arrow head extend=4pt, minimum height=5.5mm, minimum width=7mm, rotate=0] at (2.8,4.4) {};
\node[single arrow, draw=black, single arrow head extend=4pt, minimum height=5.5mm, minimum width=7mm, rotate=0] at (2.8,5.2) {};

\node[single arrow, draw=black, single arrow head extend=4pt, minimum height=5.5mm, minimum width=7mm, rotate=90] at (3.6,0.4) {};
\node[single arrow, draw=black, single arrow head extend=4pt, minimum height=5.5mm, minimum width=7mm, rotate=90] at (3.6,1.2) {};
\node[single arrow, draw=black, single arrow head extend=4pt, minimum height=5.5mm, minimum width=7mm, rotate=0] at (3.6,2.0) {};
\node[single arrow, draw=black, single arrow head extend=4pt, minimum height=5.5mm, minimum width=7mm, rotate=270] at (3.6,2.8) {};
\node[single arrow, draw=black, single arrow head extend=4pt, minimum height=5.5mm, minimum width=7mm, rotate=270] at (3.6,3.6) {};
\node[single arrow, draw=black, single arrow head extend=4pt, minimum height=5.5mm, minimum width=7mm, rotate=270] at (3.6,4.4) {};
\node[single arrow, draw=black, single arrow head extend=4pt, minimum height=5.5mm, minimum width=7mm, rotate=0] at (3.6,5.2) {};

\node[single arrow, draw=black, single arrow head extend=4pt, minimum height=5.5mm, minimum width=7mm, rotate=0] at (4.4,0.4) {};
\node[single arrow, draw=black, single arrow head extend=4pt, minimum height=5.5mm, minimum width=7mm, rotate=0] at (4.4,1.2) {};
\node[single arrow, draw=black, single arrow head extend=4pt, minimum height=5.5mm, minimum width=7mm, rotate=0] at (4.4,2.0) {};
\node[single arrow, draw=black, single arrow head extend=4pt, minimum height=5.5mm, minimum width=7mm, rotate=0] at (4.4,2.8) {};
\node[single arrow, draw=black, single arrow head extend=4pt, minimum height=5.5mm, minimum width=7mm, rotate=0] at (4.4,3.6) {};
\node[single arrow, draw=black, single arrow head extend=4pt, minimum height=5.5mm, minimum width=7mm, rotate=0] at (4.4,4.4) {};
\node[single arrow, draw=black, single arrow head extend=4pt, minimum height=5.5mm, minimum width=7mm, rotate=0] at (4.4,5.2) {};

\node[single arrow, draw=black, single arrow head extend=4pt, minimum height=5.5mm, minimum width=7mm, rotate=0] at (5.2,0.4) {};
\node[single arrow, draw=black, single arrow head extend=4pt, minimum height=5.5mm, minimum width=7mm, rotate=0] at (5.2,1.2) {};
\node[single arrow, draw=black, single arrow head extend=4pt, minimum height=5.5mm, minimum width=7mm, rotate=0] at (5.2,2.0) {};
\node[single arrow, draw=black, single arrow head extend=4pt, minimum height=5.5mm, minimum width=7mm, rotate=180] at (5.2,2.8) {};
\node[single arrow, draw=black, single arrow head extend=4pt, minimum height=5.5mm, minimum width=7mm, rotate=180] at (5.2,3.6) {};
\node[single arrow, draw=black, single arrow head extend=4pt, minimum height=5.5mm, minimum width=7mm, rotate=180] at (5.2,4.4) {};
\node[single arrow, draw=black, single arrow head extend=4pt, minimum height=5.5mm, minimum width=7mm, rotate=180] at (5.2,5.2) {};

\node[single arrow, draw=black, single arrow head extend=4pt, minimum height=5.5mm, minimum width=7mm, rotate=0] at (6.0,0.4) {};
\node[single arrow, draw=black, single arrow head extend=4pt, minimum height=5.5mm, minimum width=7mm, rotate=0] at (6.0,1.2) {};
\node[single arrow, draw=black, single arrow head extend=4pt, minimum height=5.5mm, minimum width=7mm, rotate=0] at (6.0,2.0) {};
\node[single arrow, draw=black, single arrow head extend=4pt, minimum height=5.5mm, minimum width=7mm, rotate=0] at (6.0,2.8) {};
\node[single arrow, draw=black, single arrow head extend=4pt, minimum height=5.5mm, minimum width=7mm, rotate=0] at (6.0,3.6) {};
\node[single arrow, draw=black, single arrow head extend=4pt, minimum height=5.5mm, minimum width=7mm, rotate=0] at (6.0,4.4) {};
\node[single arrow, draw=black, single arrow head extend=4pt, minimum height=5.5mm, minimum width=7mm, rotate=0] at (6.0,5.2) {};

\node[single arrow, draw=black, single arrow head extend=4pt, minimum height=5.5mm, minimum width=7mm, rotate=0] at (6.8,0.4) {};
\node[single arrow, draw=black, single arrow head extend=4pt, minimum height=5.5mm, minimum width=7mm, rotate=0] at (6.8,1.2) {};
\node[single arrow, draw=black, single arrow head extend=4pt, minimum height=5.5mm, minimum width=7mm, rotate=0] at (6.8,2.0) {};
\node[single arrow, draw=black, single arrow head extend=4pt, minimum height=5.5mm, minimum width=7mm, rotate=0] at (6.8,2.8) {};
\node[single arrow, draw=black, single arrow head extend=4pt, minimum height=5.5mm, minimum width=7mm, rotate=0] at (6.8,3.6) {};
\node[single arrow, draw=black, single arrow head extend=4pt, minimum height=5.5mm, minimum width=7mm, rotate=0] at (6.8,4.4) {};
\node[single arrow, draw=black, single arrow head extend=4pt, minimum height=5.5mm, minimum width=7mm, rotate=0] at (6.8,5.2) {};

\node[single arrow, draw=black, single arrow head extend=4pt, minimum height=5.5mm, minimum width=7mm, rotate=180] at (7.6,0.4) {};
\node[single arrow, draw=black, single arrow head extend=4pt, minimum height=5.5mm, minimum width=7mm, rotate=180] at (7.6,1.2) {};
\node[single arrow, draw=black, single arrow head extend=4pt, minimum height=5.5mm, minimum width=7mm, rotate=180] at (7.6,2.0) {};
\node[single arrow, draw=black, single arrow head extend=4pt, minimum height=5.5mm, minimum width=7mm, rotate=180] at (7.6,2.8) {};
\node[single arrow, draw=black, single arrow head extend=4pt, minimum height=5.5mm, minimum width=7mm, rotate=180] at (7.6,3.6) {};
\node[single arrow, draw=black, single arrow head extend=4pt, minimum height=5.5mm, minimum width=7mm, rotate=180] at (7.6,4.4) {};
\node[single arrow, draw=black, single arrow head extend=4pt, minimum height=5.5mm, minimum width=7mm, rotate=180] at (7.6,5.2) {};

\node at (0.4,0.4) {2};
\node at (0.4,1.2) {2};
\node at (0.4,2.0) {4};
\node at (0.4,2.8) {4};
\node at (0.4,3.6) {2};
\node at (0.4,4.4) {4};
\node at (0.4,5.2) {2};

\node at (1.2,0.4) {3};
\node at (1.2,1.2) {2};
\node at (1.2,2.0) {4};
\node at (1.2,2.8) {3};
\node at (1.2,3.6) {4};
\node at (1.2,4.4) {3};
\node at (1.2,5.2) {3};

\node at (2.0,0.4) {2};
\node at (2.0,1.2) {2};
\node at (2.0,2.0) {4};
\node at (2.0,2.8) {4};
\node at (2.0,3.6) {4};
\node at (2.0,4.4) {2};
\node at (2.0,5.2) {2};

\node at (2.8,0.4) {2};
\node at (2.8,1.2) {2};
\node at (2.8,2.0) {4};
\node at (2.8,2.8) {2};
\node at (2.8,3.6) {2};
\node at (2.8,4.4) {4};
\node at (2.8,5.2) {4};

\node at (3.6,0.4) {3};
\node at (3.6,1.2) {2};
\node at (3.6,2.0) {4};
\node at (3.6,2.8) {3};
\node at (3.6,3.6) {3};
\node at (3.6,4.4) {3};
\node at (3.6,5.2) {4};

\node at (4.4,0.4) {4};
\node at (4.4,1.2) {4};
\node at (4.4,2.0) {4};
\node at (4.4,2.8) {4};
\node at (4.4,3.6) {4};
\node at (4.4,4.4) {4};
\node at (4.4,5.2) {4};

\node at (5.2,0.4) {3};
\node at (5.2,1.2) {3};
\node at (5.2,2.0) {3};
\node at (5.2,2.8) {3};
\node at (5.2,3.6) {3};
\node at (5.2,4.4) {3};
\node at (5.2,5.2) {3};

\node at (6.0,0.4) {2};
\node at (6.0,1.2) {2};
\node at (6.0,2.0) {2};
\node at (6.0,2.8) {2};
\node at (6.0,3.6) {2};
\node at (6.0,4.4) {2};
\node at (6.0,5.2) {2};

\node at (6.8,0.4) {1};
\node at (6.8,1.2) {1};
\node at (6.8,2.0) {1};
\node at (6.8,2.8) {1};
\node at (6.8,3.6) {1};
\node at (6.8,4.4) {1};
\node at (6.8,5.2) {1};

\node at (7.6,0.4) {4};
\node at (7.6,1.2) {4};
\node at (7.6,2.0) {4};
\node at (7.6,2.8) {4};
\node at (7.6,3.6) {4};
\node at (7.6,4.4) {4};
\node at (7.6,5.2) {4};
\end{tikzpicture}
\caption{A solution of the Zeiger puzzle in Fig.~\ref{fig3} transformed from the NAE3SAT+ solution in Fig.~\ref{fig2}}
\label{fig4}
\end{figure}

From Theorem \ref{thm1}, we can conclude that deciding solvability of a given Zeiger puzzle is NP-complete.

\section{Physical ZKP Protocol for Zeiger}
Because of the NP-completeness of Zeiger, it is worth proposing a ZKP for it. In this section, we will construct a card-based physical ZKP protocol for the puzzle.

\subsection{Preliminaries}
\subsubsection{Cards}
Each card used in our protocol has a $\clubsuit$ or a $\heartsuit$ written on the front side. All cards have indistinguishable back sides denoted by \mybox{?}.

For $0 \leq x < q$, define $E_q^\clubsuit(x)$ to be a sequence of $q$ cards, all of them being \hbox{\mybox{$\heartsuit$}s} except the $(x+1)$-th leftmost card being a \mybox{$\clubsuit$}. For example, $E_4^\clubsuit(1)$ is \hbox{\mybox{$\heartsuit$}\mybox{$\clubsuit$}\mybox{$\heartsuit$}\mybox{$\heartsuit$}}. Analogously, define $E_q^\heartsuit(x)$ to be a sequence of $q$ cards, all of them being \hbox{\mybox{$\clubsuit$}s} except the $(x+1)$-th leftmost card being a \mybox{$\heartsuit$}. For example, $E_4^\heartsuit(1)$ is \hbox{\mybox{$\clubsuit$}\mybox{$\heartsuit$}\mybox{$\clubsuit$}\mybox{$\clubsuit$}}.

We may stack $E_q^\clubsuit(x)$ or $E_q^\heartsuit(x)$ into a single stack, with the leftmost card being the topmost card in the stack. Note that we refer to the ``sequence form'' and ``stack form'' of them interchangeably throughout the protocol.

Also, for $0 \leq x < q$, define $E_q(x)$ to be a sequence of $q$ two-card stacks, all of them being \hbox{$E_2^\clubsuit(0)$s} except the $(x+1)$-th leftmost stack being an $E_2^\clubsuit(1)$. For example, $E_4(1)$ is \hbox{\mybox{$\clubsuit$}\mybox{$\heartsuit$} \mybox{$\heartsuit$}\mybox{$\clubsuit$} \mybox{$\clubsuit$}\mybox{$\heartsuit$} \mybox{$\clubsuit$}\mybox{$\heartsuit$}} (when extracting cards in each stack into a sequence). Note that the top cards of all stacks of $E_q(x)$ form a sequence $E_q^\heartsuit(x)$, and the bottom cards of all stacks of $E_q(x)$ form a sequence $E_q^\clubsuit(x)$.

\subsubsection{Pile-Shifting Shuffle}
A \textit{pile-shifting shuffle} \cite[\S2.3]{polygon} rearranges all columns of a matrix of cards (or a matrix of stacks) by a uniformly random cyclic shift unknown to $V$. It can be implemented by putting all cards in each column into an envelope, and repeatedly picking some envelopes from the bottom and putting them on the top of the pile of envelopes.

\subsubsection{Pile-Scramble Shuffle}
A \textit{pile-scramble shuffle} \cite[\S3]{scramble} rearranges all columns of a matrix of cards (or a matrix of stacks) by a uniformly random permutation unknown to $V$. It can be implemented by putting all cards in each column into an envelope, and scrambling all envelopes together on a table.

\subsubsection{Copy Protocol}
Given a sequence $E_q(x)$, a \textit{copy protocol} creates an additional copy of the sequence without revealing $x$ to $V$. This protocol was developed by Shinagawa et al. \cite[\S3.1]{polygon} (using different encoding).

\begin{algorithm}[H] \label{sub1}
    \floatname{algorithm}{Subprotocol}
    \caption{Copy Protocol}
    \textbf{Input:} a face-down sequence $A = E_q(x)$ and $4q$ extra cards ($2q$ \mybox{$\clubsuit$}s and $2q$ \mybox{$\heartsuit$}s) \\
		\textbf{Output:} two face-down sequences $E_q(x)$ and $E_q(x)$ \\
    \textbf{Procedures:} $P$ performs the following steps.
    \begin{enumerate}
        \item Reverse the $q-1$ rightmost stacks of $A$, i.e. move the $(i+1)$-th leftmost stack to become the $i$-th rightmost stack for each $i=1,2,...,q-1$. The sequence $A$ now becomes $E_q(-x \mod q)$.
				\item Publicly construct two $E_q(0)$s from the $4q$ extra cards.
				\item Construct a $3 \times q$ matrix of two-card stacks $M$ by placing $A$ in Row 1 and the two $E_q(0)$s in Rows 2 and 3. Turn all cards face-down.
				\item Apply the pile-shifting shuffle to $M$.
				\item Turn over all cards in Row 1 of $M$. Shift the columns of $M$ cyclically such that the only $E_2^\clubsuit(1)$ in Row 1 moves to Column 1.
				\item Return the two sequences in Rows 2 and 3 of $M$.
    \end{enumerate}
\end{algorithm}

In Step 3, the sequence in Row 1 of $M$ is $E_q(-x \mod q)$, and the sequences in Rows 2 and 3 of $M$ are both $E_q(0)$s. Observe that the operations in Steps~4 and~5 only shift the columns of $M$ cyclically, so the numbers encoded by sequences in all rows always shift together by the same constant modulo $q$. In Step 6, the sequence in Row 1 becomes an $E_q(0)$, thus the sequences in Rows 2 and 3 must both be $E_q(x)$s.

Note that this protocol also checks that the sequence $A$ is in the correct format, i.e. consists of $q-1$ $E_2^\clubsuit(0)$s and one $E_2^\clubsuit(1)$. In Step 5, if the sequence in Row~1 does not consist of $q-1$ $E_2^\clubsuit(0)$s and one $E_2^\clubsuit(1)$, $V$ immediately rejects the verification.

\subsubsection{Set Size Protocol}
Given $p$ integers $x_1,x_2,...,x_p \in \{0,1,...,q-1\}$ unknown to $V$, a \textit{set size protocol} computes the number of different values among $x_1,x_2,...,x_p$, i.e. the size $\left|\{x_1,x_2,...,x_p\}\right|$, without revealing any $x_i$ to $V$.

This protocol was developed by Ruangwises et al. \cite[\S5.1]{equality2} (using different encoding), and uses a similar idea to the set union protocol of Doi et al. \cite[\S5.1.2]{setunion}.

\begin{algorithm}[H] \label{sub2}
    \floatname{algorithm}{Subprotocol}
    \caption{Set Size Protocol}
    \textbf{Input:} $p$ face-down sequences $E_q(x_1),E_q(x_2),...,E_q(x_p)$ \\
		\textbf{Output:} $q$ face-down sequences $E_2^\clubsuit(y_1),E_2^\clubsuit(y_2),...,E_2^\clubsuit(y_q)$ such that $y_1,y_2,...,y_q \in \{0,1\}$ and $y_1+y_2+...+y_q=|\{x_1,x_2,...,x_p\}|$ \\
    \textbf{Procedures:} $P$ performs the following steps.
    \begin{enumerate}
			\item Construct a $p \times q$ matrix of two-card stacks $M$ by placing $E_q(x_i)$ in Row $i$ of $M$ for each $i=1,2,...,p$. Let $M(i,j)$ denotes the stack at Row $i$ and Column $j$ of $M$.
			\item Perform the following steps for $i=2,3,...,p$.
			\begin{enumerate}
				\item Apply the pile-scramble shuffle to $M$.
				\item Turn over all cards in Row $i$ of $M$.
				\item For each $j=1,2,...,q$, if $M(i,j)$ is an $E_2^\clubsuit(1)$, swap $M(i,j)$ and $M(1,j)$.
				\item Turn all cards face-down.
			\end{enumerate}
			\item Apply the pile-scramble shuffle to $M$.
			\item Return the $q$ stacks in Row 1 of $M$ as $E_2^\clubsuit(y_1),E_2^\clubsuit(y_2),...,E_2^\clubsuit(y_q)$.
		\end{enumerate}
\end{algorithm}

Observe that in Step 2(c), we only replace $M(1,j)$ with an $E_2^\clubsuit(1)$. Therefore, once a stack in Row 1 becomes an $E_2^\clubsuit(1)$, it will remain an $E_2^\clubsuit(1)$ throughout the protocol.

For any $r \in \{1,2,...,q\}$, consider the stack $M(1,r)$ at the beginning (we consider this exact stack no matter which column it later moves to). If it is an $E_2^\clubsuit(1)$ at the beginning, i.e. $x_1=r-1$, it will remain an $E_2^\clubsuit(1)$ at the end. On the other hand, if it is an $E_2^\clubsuit(0)$ at the beginning, it will become an $E_2^\clubsuit(1)$ at the end if and only if it has been replaced by an $E_2^\clubsuit(1)$ from some Row $i \geq 2$ during the $i$-th iteration, which occurs when $x_i=r-1$.

Therefore, the stack $M(1,r)$ at the beginning is an $E_2^\clubsuit(1)$ at the end if and only if there is $i \in \{1,2,\ldots ,p\}$ such that $x_i=r-1$. This implies the number of $E_2^\clubsuit(1)$s in Row 1 at the end is equal to the number of different values among $x_1,x_2,...,x_p$.

\subsubsection{Summation Protocol}
Given $q$ integers $y_1,y_2,...,y_q \in \{0,1\}$ unknown to $V$, a~\textit{summation protocol} computes the sum $y_1+y_2+...+y_q$ without revealing any $y_i$ to $V$. This protocol was developed by Ruangwises and Itoh \cite[\S3.1]{equality}.

\begin{algorithm}[H] \label{sub3}
    \floatname{algorithm}{Subprotocol}
    \caption{Summation Protocol}
    \textbf{Input:} $q$ face-down sequences $E_2^\clubsuit(y_1),E_2^\clubsuit(y_2),...,E_2^\clubsuit(y_q)$ and two extra cards (a \mybox{$\clubsuit$} and a \mybox{$\heartsuit$}) \\
		\textbf{Output:} a face-down sequence $E_{q+1}^\clubsuit(y_1+y_2+...+y_q)$ \\
    \textbf{Procedures:} $P$ performs the following steps.
    \begin{enumerate}
        \item Let $A$ be the sequence $E_2^\clubsuit(y_1)$.
        \item Perform the following steps for $i=2,3,...,q$.
				\begin{enumerate}
					\item Append a \mybox{$\heartsuit$} to the right of $A$, making it an $E_{i+1}^\clubsuit(y_1+y_2+...+y_{i-1})$.
					\item Swap the two cards of $E_2^\clubsuit(y_i)$, making it $E_2^\heartsuit(y_i)$. Then, append $i-1$ \mybox{$\clubsuit$}s between the two cards, making it $E_{i+1}^\heartsuit(-y_i \mod (i+1))$. Call this sequence~$B$.
					\item Construct a $2 \times (i+1)$ matrix of cards $M$ by placing $A$ in Row 1 and $B$ in Row 2. Turn all cards face-down.
					\item Apply the pile-shifting shuffle to $M$.
					\item Turn over all cards in Row 2 of $M$. Shift the columns of $M$ cyclically such that the only \mybox{$\heartsuit$} in Row 2 moves to Column 1.
					\item Let $A$ be the sequence in Row 1 of $M$, to be used in the next iteration. Note that $A$ is $E_{i+1}^\clubsuit(y_1+y_2+...+y_i)$.
				\end{enumerate}
				\item Return the sequence $A$ from the final iteration.
    \end{enumerate}
\end{algorithm}

The purpose of the $i$-th iteration of Step 2 is to add $y_1+y_2+...+y_{i-1}$ and $y_i$. In Step 2(c), the sequence in Row 1 of $M$ is $E_{i+1}^\clubsuit(y_1+y_2+...+y_{i-1})$, and the sequence in Row 2 of $M$ is $E_{i+1}^\heartsuit(-y_i \mod (i+1))$. Observe that the operations in Steps 2(d) and 2(e) only shift the columns of $M$ cyclically, so the numbers encoded by sequences in both rows always shift together by the same constant modulo $q$. In Step 2(f), the sequence in Row 2 becomes $E_{i+1}^\heartsuit(0)$, thus the sequence in Row 1 must be $E_{i+1}^\clubsuit(y_1+y_2+...+y_i)$.

Moreover, at the end of the $i$-th iteration, we get $i+1$ extra cards ($i$ \mybox{$\clubsuit$}s and one \mybox{$\heartsuit$}) from the unused sequence in Row 2 of $M$ to use in Steps 2(a) and 2(b) in the next iteration without requiring additional cards.

After the $q$-th iteration, the sequence $A$ will become $E_{q+1}^\clubsuit(y_1+y_2+...+y_q)$.

\subsubsection{Comparing Protocol}
Given two integers $x_1,x_2 \in \{0,1,...,q-1\}$ unknown to $V$, a \textit{comparing protocol} lets $V$ verify that $x_1=x_2$ without revealing their value to $V$. This protocol was developed by Bultel et al. \cite[\S3.3]{makaro}.

\begin{algorithm}[H] \label{sub4}
    \floatname{algorithm}{Subprotocol}
    \caption{Comparing Protocol}
    \textbf{Input:} two face-down sequences $E_q^\clubsuit(x_1)$ and $E_q^\clubsuit(x_2)$ \\
		\textbf{Output:} $V$ rejects the verification if $x_1 \neq x_2$, or does nothing if $x_1=x_2$. \\
    \textbf{Procedures:} $P$ performs the following steps.
    \begin{enumerate}
        \item Construct a $2 \times q$ matrix of cards $M$ by placing $E_q^\clubsuit(x_1)$ and $E_q^\clubsuit(x_2)$ in Row 1 and Row 2 of $M$, respectively.
        \item Apply the pile-scramble shuffle to $M$.
				\item Turn over all cards in $M$. If the \mybox{$\clubsuit$} in Row 1 and the \mybox{$\clubsuit$} in Row 2 are not in the same column, $V$ rejects.
    \end{enumerate}
\end{algorithm}

Observe that we only rearrange the columns of $M$, so the \mybox{$\clubsuit$}s in both rows are in the same column at the end if and only if they were in the same column at the beginning, i.e. $x_1=x_2$.

\subsection{Our Main Protocol}
Suppose the puzzle grid has size $k \times \ell$. Let $b = \max\{k,\ell\}$. Observe that an integer in each cell in the puzzle grid can be at most $b-1$.

On each cell with an integer $x$, $P$ publicly places a face-up $E_b(x)$. On each unnumbered cell, $P$ secretly places a face-down $E_b(x)$, where $x$ is an integer in that cell in $P$'s solution. Then, $P$ turns all cards face-down.

For each cell $c$, $P$ performs the following steps to check the constraint on $c$.
\begin{enumerate}
	\item Let $c_1,c_2,...,c_t$ be all cells along the direction the arrow in $c$ points to. Let $E_b(d)$ be the sequence on $c$, and $E_b(d_1),E_b(d_2),...,E_b(d_t)$ be the sequences on $c_1,c_2,...,c_t$, respectively.
	\item Apply the copy protocol to make a copy of each of $E_b(d),E_b(d_1),E_b(d_2),...,$ $E_b(d_t)$.
	\item Apply the set size protocol to the copies of $E_b(d_1),E_b(d_2),...,E_b(d_t)$ to get outputs $E_2^\clubsuit(z_1),E_2^\clubsuit(z_2),...,E_2^\clubsuit(z_b)$, where $z_1+z_2+...+z_b=|\{d_1,d_2,...,d_t\}|$.
	\item Apply the summation protocol to $E_2^\clubsuit(z_1),E_2^\clubsuit(z_2),...,E_2^\clubsuit(z_b)$ to get output $E_{b+1}^\clubsuit(z)$, where $z=z_1+z_2+...+z_b$. Note that $z$ is the number of different integers in $c_1,c_2,...,c_t$.
	\item Append a \mybox{$\heartsuit$} to the right of the copy of $E_b^\clubsuit(d)$ from Step 2 to make it an $E_{b+1}^\clubsuit(d)$. Then, apply the comparing protocol to $E_{b+1}^\clubsuit(d)$ and $E_{b+1}^\clubsuit(z)$ to check that $d=z$.
\end{enumerate}

$P$ performs these steps for every cell in the grid. If all verifications pass, then $V$ accepts.

Our protocol uses $\Theta(bk\ell)$ cards and $\Theta(bk\ell)$ shuffles.

\subsection{Proof of Correctness and Security}
We will prove the perfect completeness, perfect soundness, and zero-knowledge properties of our main protocol.

\begin{lemma}[Perfect Completeness] \label{lem1}
If $P$ knows a solution of the Zeiger puzzle, then $V$ always accepts.
\end{lemma}

\begin{proof}
Suppose $P$ knows a solution of the puzzle. $P$ can place sequences on all cells according to the solution.

Consider the verification of each cell $c$ with a sequence $E_b(d)$. Because of the constraint of the puzzle, $d$ is equal to the number of different integers in cells $c_1,c_2,...,c_t$, i.e. $d=z$. Therefore, the verification for $c$ will pass.

Since this is true for every cell, we can conclude that $V$ always accepts.
\end{proof}

\begin{lemma}[Perfect Soundness] \label{lem2}
If $P$ does not know a solution of the Zeiger puzzle, then $V$ always rejects.
\end{lemma}

\begin{proof}
Suppose $P$ does not know a solution of the puzzle. First, if a sequence on some cell does not have the correct format, i.e. does not consist of $b-1$ $E_2^\clubsuit(0)$s and one $E_2^\clubsuit(1)$, it will be detected by the copy protocol, and $V$ will immediately reject.

Suppose all sequences have the correct format. Since $P$ does not know a solution, the sequence on at least one cell, say $c$, must violate the constraint of the puzzle. During the verification for the cell $c$, we will have $d \neq z$, and $V$ will reject during the comparing protocol.

Hence, we can conclude that $V$ always rejects.
\end{proof}

\begin{lemma}[Zero-Knowledge] \label{lem3}
During the verification, $V$ obtains no information about $P$'s solution.
\end{lemma}

\begin{proof}
We will prove that any interaction between $P$ and $V$ can be simulated by a simulator $S$ that does not know $P$'s solution. It is sufficient to show that all distributions of cards that are turned face-up can be simulated by $S$.

\begin{itemize}
	\item In Step 5 of the copy protocol, because of the pile-shifting shuffle, the $E_2^\clubsuit(1)$ has probability $1/q$ to be at each of the $q$ columns. Therefore, this step can be simulated by $S$.
	\item In Step 2(b) in each $i$-th iteration of the set size protocol, because of the pile-scramble shuffle, the $E_2^\clubsuit(1)$ has probability $1/q$ to be at each of the $q$ columns. Therefore, this step can be simulated by $S$.
	\item In Step 2(e) in each $i$-th iteration of the summation protocol, because of the pile-shifting shuffle, the \mybox{$\heartsuit$} has probability $1/(i+1)$ to be at each of the $i+1$ columns. Therefore, this step can be simulated by $S$.
	\item In Step 3 of the set size protocol (in the case that $V$ does not reject), because of the pile-scramble shuffle, the \mybox{$\clubsuit$}s in both rows are in the same column and have probability $1/q$ to be at each of the $q$ columns. Therefore, this step can be simulated by $S$.
\end{itemize}

Hence, we can conclude that $V$ obtains no information about $P$'s solution.
\end{proof}

\section{Future Work}
We proved the NP-completeness of a pencil puzzle Zeiger and constructed a card-based physical ZKP protocol for it. Future work includes proving the NP-completeness of other well-known pencil puzzles, as well as constructing card-based ZKP protocols for them.


\begin{thebibliography}{99}
	\bibitem{npnumberlink} A. Adcock, E.D. Demaine, M.L. Demaine, M.P. O'Brien, F. Reidl, F.S. Villaamil and B.D. Sullivan. Zig-Zag Numberlink is NP-Complete. \textit{Journal of Information Processing}, 23(3): 239--245 (2015).
	\bibitem{npbridges} D. Andersson. Hashiwokakero is NP-complete. \textit{Information Processing Letters}, 109(9): 1145--1146 (2009).
	\bibitem{npgoishi} D. Andersson. HIROIMONO Is NP-Complete. In \textit{Proceedings of the 4th International Conference on Fun with Algorithms (FUN)}, pp. 30--39 (2007).
	\bibitem{apps} App Store: Number Pointers. \url{https://apps.apple.com/us/app/number-pointers/id1065321061}
	\bibitem{makaro} X. Bultel, J. Dreier, J.-G. Dumas, P. Lafourcade, D. Miyahara, T. Mizuki, A. Nagao, T. Sasaki, K. Shinagawa and H. Sone. Physical Zero-Knowledge Proof for Makaro. In \textit{Proceedings of the 20th International Symposium on Stabilization, Safety, and Security of Distributed Systems (SSS)}, pp. 111--125 (2018).
	\bibitem{setunion} A. Doi, T. Ono, Y. Abe, T. Nakai, K. Shinagawa, Y. Watanabe, K. Nuida and M. Iwamoto. Card-Based Protocols for Private Set Intersection and Union. \textit{New Generation Computing}, 42(3): 359--380 (2024).
	\bibitem{zkp} O. Goldreich, S. Micali and A. Wigderson. Proofs that yield nothing but their validity and a methodology of cryptographic protocol design. \textit{Journal of the ACM}, 38(3): 691--729 (1991).
	\bibitem{zkp0} S. Goldwasser, S. Micali and C. Rackoff. The knowledge complexity of interactive proof systems. \textit{SIAM Journal on Computing}, 18(1): 186--208 (1989).
	\bibitem{sumplete} K. Hatsugai, S. Ruangwises, K. Asano and Y. Abe. NP-Completeness and Physical Zero-Knowledge Proofs for Sumplete, a Puzzle Generated by ChatGPT. \textit{New Generation Computing}, 42(3): 429--448 (2024).
	\bibitem{npnurikabe} M. Holzer, A. Klein and M. Kutrib. On The NP-Completeness of The Nurikabe Pencil Puzzle and Variants Thereof. In \textit{Proceedings of the 3rd International Conference on Fun with Algorithms (FUN)}, pp. 77--89 (2004).
	\bibitem{npheyawake} M. Holzer and O. Ruepp. The Troubles of Interior Design–A Complexity Analysis of the Game Heyawake. In \textit{Proceedings of the 4th International Conference on Fun with Algorithms (FUN)}, pp. 198--212 (2007).
	\bibitem{scramble} R. Ishikawa, E. Chida and T. Mizuki. Efficient Card-Based Protocols for Generating a Hidden Random Permutation Without Fixed Points. In \textit{Proceedings of the 14th International Conference on Unconventional Computation and Natural Computation (UCNC)}, pp. 215--226 (2015).
	\bibitem{npmakaro} C. Iwamoto, M. Haruishi and T. Ibusuki. Herugolf and Makaro are NP-complete. In \textit{Proceedings of the 9th International Conference on Fun with Algorithms (FUN)}, pp. 24:1--24:11 (2018).
	\bibitem{npfivecells} C. Iwamoto and T. Ide. Five Cells and Tilepaint are NP-Complete. \textit{IEICE Trans. Inf. \& Syst.}, 105.D(3): 508--516 (2022).
	\bibitem{janko} A. Janko and O. Janko. Zeiger (Logikrätsel). \url{https://www.janko.at/Raetsel/Zeiger/index.htm}
	\bibitem{slitherlink} P. Lafourcade, D. Miyahara, T. Mizuki, L. Robert, T. Sasaki and H. Sone. How to construct physical zero-knowledge proofs for puzzles with a ``single loop'' condition. \textit{Theoretical Computer Science}, 888: 41--55 (2021).
	\bibitem{kakuro} D. Miyahara, T. Sasaki, T. Mizuki and H. Sone. Card-Based Physical Zero-Knowledge Proof for Kakuro. \textit{IEICE Trans. Fundamentals}, E102.A(9): 1072--1078 (2019).
	\bibitem{nurikabe} L. Robert, D. Miyahara, P. Lafourcade and T. Mizuki. Card-Based ZKP for Connectivity: Applications to Nurikabe, Hitori, and Heyawake. \textit{New Generation Computing}, 40(1): 149--171 (2022).
	\bibitem{nonogram} S. Ruangwises. An Improved Physical ZKP for Nonogram and Nonogram Color. \textit{Journal of Combinatorial Optimization}, 45(5): 122 (2023).
	\bibitem{npnondango} S. Ruangwises. Nondango is NP-Complete. In \textit{Proceedings of the 40th European Workshop on Computational Geometry (EuroCG)}, pp. 1:1--1:10 (2024).
	\bibitem{goishi} S. Ruangwises. Verifying the First Nonzero Term: Physical ZKPs for ABC End View, Goishi Hiroi, and Toichika. \textit{Journal of Combinatorial Optimization}, 47(4): 69 (2024).
	\bibitem{shikaku} S. Ruangwises and T. Itoh. How to Physically Verify a Rectangle in a Grid: A Physical ZKP for Shikaku. In \textit{Proceedings of the 11th International Conference on Fun with Algorithms (FUN)}, pp. 24:1--24:12 (2022).
	\bibitem{numberlink} S. Ruangwises and T. Itoh. Physical Zero-Knowledge Proof for Numberlink Puzzle and $k$ Vertex-Disjoint Paths Problem. \textit{New Generation Computing}, 39(1): 3--17 (2021).
	\bibitem{ripple} S. Ruangwises and T. Itoh. Physical Zero-Knowledge Proof for Ripple Effect. \textit{Theoretical Computer Science}, 895: 115--123 (2021).
	\bibitem{equality} S. Ruangwises and T. Itoh. Securely Computing the $n$-Variable Equality Function with $2n$ Cards. \textit{Theoretical Computer Science}, 887: 99--110 (2021).
	\bibitem{decom} S. Ruangwises and M. Iwamoto. Printing Protocol: Physical ZKPs for Decomposition Puzzles. \textit{New Generation Computing}, 42(3): 331--343 (2024).
	\bibitem{equality2} S. Ruangwises, T. Ono, Y. Abe, K. Hatsugai and M. Iwamoto. Card-Based Overwriting Protocol for Equality Function and Applications. In \textit{Proceedings of the 21st International Conference on Unconventional Computation and Natural Computation (UCNC)}, pp. 18--27 (2024).
	\bibitem{sudoku} T. Sasaki, D. Miyahara, T. Mizuki and H. Sone. Efficient card-based zero-knowledge proof for Sudoku. \textit{Theoretical Computer Science}, 839: 135--142 (2020).
	\bibitem{sat} T.J. Schaefer. The complexity of satisfiability problems. In \textit{Proceedings of the 10th Annual ACM Symposium on Theory of Computing (STOC)}, pp. 216--226 (1978).
	\bibitem{polygon} K. Shinagawa, T. Mizuki, J.C.N. Schuldt, K. Nuida, N. Kanayama, T. Nishide, G. Hanaoka and E. Okamoto. Card-Based Protocols Using Regular Polygon Cards. \textit{IEICE Trans. Fundamentals}, E100.A(9): 1900--1909 (2017).
	\bibitem{npripple} Y. Takenaga, S. Aoyagi, S. Iwata and T. Kasai. Shikaku and Ripple Effect are NP-Complete. \textit{Congressus Numerantium}, 216: 119--127 (2013).
	\bibitem{npnonogram} N. Ueda and T. Nagao. NP-completeness Results for NONOGRAM via Parsimonious Reductions. Technical Report TR96-0008, Department of Computer Science, Tokyo Institute of Technology (1996).
	\bibitem{npsudoku} T. Yato and T. Seta. Complexity and Completeness of Finding Another Solution and Its Application to Puzzles. \textit{IEICE Trans. Fundamentals}, 86.A(5): 1052--1060 (2003).
\end{thebibliography}
\end{document}